\documentclass[twoside]{article}
\usepackage[a4paper]{geometry}
\usepackage[latin1]{inputenc} 
\usepackage[T1]{fontenc} 
\usepackage{RR}
\usepackage{hyperref}
\usepackage{url}
\usepackage{xspace}

\usepackage{amssymb}
\usepackage{amsmath}
\usepackage{amsthm}

\usepackage{array}
\usepackage{float}
\usepackage{graphicx}
\usepackage{verbatim}
\usepackage{algorithm}
\usepackage[noend]{algorithmic}
\usepackage{comment}
\usepackage{soul}
\usepackage{color}

\usepackage{authblk}
\usepackage{longtable}
\usepackage{multirow}

\restylefloat{figure}
\restylefloat{table}

\usepackage{algorithm}
\usepackage[noend]{algorithmic}

\usepackage{graphicx}
\usepackage{color}
\usepackage{algorithm}
\usepackage[noend]{algorithmic}
\usepackage{pdfpages}

\usepackage{amsmath}
\usepackage{amssymb}
\usepackage{latexsym}
\usepackage{array}
\usepackage{stmaryrd}
\usepackage{paralist}
\usepackage{graphicx} 
\usepackage{xspace}
\usepackage{url}
\usepackage{tabulary,booktabs,multirow}
\usepackage{algorithm, algorithmic}
\usepackage{tikz}
\usepackage{footnote}
\usepackage{subfig}
\usepackage[shortlabels]{enumitem}
\setlist{nolistsep}

\newcommand{\ie}{\emph{i.e.,}\xspace}
\newcommand{\eg}{\emph{e.g.}\xspace}

\DeclareMathOperator{\vc}{vc}
\DeclareMathOperator{\ov}{over}
\DeclareMathOperator{\enc}{encomp}

\newcommand{\N}{\mathbb{N}}

\renewcommand{\S}{\mathcal{S}}

\newtheorem{theorem}{Theorem}
\newtheorem{corollary}{Corollary}

\newtheorem{lemma}[theorem]{Lemma}

\newcommand{\K}{\overline{K}}

\newcommand{\F}{\mathcal{F}}

\newcommand{\fpt}{{\sf FPT}\xspace}
\newcommand{\W}{{\sf W}\xspace}
\newcommand{\np}{{\sf NP}\xspace}
\newcommand{\poly}{{\sf P}\xspace}

\newcommand{\computationalproblem}[3]{
  \vspace{2mm}
\noindent\phantom{ee} {\sc #1}\\[1mm]
\phantom{ee} \parbox{1.5cm}{\underline{Input:}} #2\\
\phantom{ee} \parbox{1.5cm}{\underline{Question:}} #3\\
}

\newif\ifLONG
\LONGtrue

\RRNo{9045}
\RRdate{Mars 2017}

\RRauthor{Nathann Cohen$^{1}$, Fr\'{e}d\'{e}ric Havet$^{2}$, Dorian Mazauric$^{3}$, Ignasi Sau$^{4,5}$ and R\'{e}mi Watrigant$^{3}$ \\
~ \\
\vspace{0.05cm}
\small{1. Universit\'e Paris-Sud, LRI, CNRS, Orsay, France\\
2. Universit\'e C\^ote d'Azur, CNRS, I3S, Inria, Sophia Antipolis, France. \\
3. Universit\'e C\^ote d'Azur, Inria, Sophia Antipolis, France. \\
4. CNRS, LIRMM, Universit\'e de Montpellier, Montpellier, France. \\
5. Departamento de Matem\'atica, Universidade Federal do Cear\'a, Fortaleza, Brazil. \\
}}

\authorhead{Cohen / Havet / Mazauric / Sau / Watrigant}
\RRtitle{Dichotomies de complexit\'e pour le probl\`eme \textsc{Minimum $\F$-Overlay}}
\RRetitle{Complexity Dichotomies for the \textsc{Minimum $\F$-Overlay} Problem\footnote{This work was partially funded by `Projet de Recherche Exploratoire', Inria, \textit{Improving inference algorithms for macromolecular structure determination} and ANR under contract STINT ANR-13-BS02-0007.}}
\titlehead{Complexity Dichotomies for the \textsc{Minimum $\F$-Overlay} Problem}
\RRresume{Pour une famille (possiblement infinie) de graphes fix\'es $\F$, un graphe $G$ \emph{couvre} $\F$ dans un hypergraphe $H$ si $V(H)$ est \'egal \`a $V(G)$ et le sous-graphe de $G$ induit par chaque hyperar\^ete de $H$ contient un membre de $\F$ comme sous-graphe couvrant.
Il est facile d'observer que le graphe complet sur $|V(H)|$ couvre $\F$ dans un hypergraphe $H$ d\`es lors qu'il y a une solution, le probl\`eme \textsc{Minimum $\F$-Overlay} consiste alors \`a calculer un tel graphe avec le nombre minimum d'ar\^etes.
Ce probl\`eme g\'en\'eralise certains autres qui ont une application en pratique. Par exemple, si la famille $\F$ contient tous les graphes connexes, alors \textsc{Minimum $\F$-Overlay} correspond au probl\`eme \textsc{Minimum Connectivity Inference} (aussi connu comme le probl\`eme \textsc{Subset Interconnection Design}) introduit pour la reconstruction basse-r\'esolution d'un assemblage macro-mol\'eculaire en biologie structurale, ou pour la conception de r\'eseaux.

Notre principale contribution est un r\'esultat de dichotomie forte concernant le statut polynomial versus \np-difficile par rapport \`a la famille $\F$ consid\'er\'e. En termes simples, nous montrons que les cas faciles (\textit{e.g.} quand les graphes sans ar\^ete de bonnes tailles sont tous dans $\F$ ou si $\F$ contient seulement des cliques) sont les seules familles qui rendent le probl\`eme polynomial \`a r\'esoudre : tous les autres sont  \np-complets.
Nous analysons ensuite la complexit\'e param\'etr\'ee du probl\`eme et prouvons des conditions suffisantes sur $\F$ qui montrent que le probl\`eme est $\W[1]$-hard, $\W[2]$-hard ou $\fpt$ quand le param\`etre est la taille de la solution.
Cela donne une dichotomie \fpt/$\W[1]$-difficile pour une version rel\^ach\'ee du probl\`eme pour laquelle chaque hyperar\^ete de $H$ doit contenir un membre de $\F$ comme sous-graphe (pas n\'ecessairement couvrant).
}

\RRabstract{
 For a (possibly infinite) fixed family of graphs $\F$, we say that a graph $G$ \emph{overlays} $\F$ on a hypergraph $H$ if $V(H)$ is equal to $V(G)$ and the subgraph of $G$ induced by every hyperedge of $H$ contains some member of $\F$ as a spanning subgraph.
  While it is easy to see that the complete graph on $|V(H)|$ overlays $\F$ on a hypergraph $H$ whenever the problem admits a solution, the \textsc{Minimum $\F$-Overlay} problem asks for such a graph with the minimum number of edges.
This problem allows to generalize some natural problems which may arise in practice. For instance, if the family $\F$ contains all connected graphs, then \textsc{Minimum $\F$-Overlay} corresponds to the \textsc{Minimum Connectivity Inference} problem (also known as \textsc{Subset Interconnection Design} problem) introduced for the low-resolution reconstruction of macro-molecular assembly in structural biology, or for the design of networks.

Our main contribution is a strong dichotomy result regarding the polynomial vs. \np-hard status with respect to the considered family $\F$. Roughly speaking, we show that the easy cases one can think of (\textit{e.g.} when edgeless graphs of the right sizes are in $\F$, or if $\F$ contains only cliques) are the only families giving rise to a polynomial problem: all others are \np-complete.
We then investigate the parameterized complexity of the problem and give similar sufficient conditions on $\F$ that give rise to $\W[1]$-hard, $\W[2]$-hard or $\fpt$ problems when the parameter is the size of the solution.
 This yields an \fpt/$\W[1]$-hard dichotomy for a relaxed problem, where every hyperedge of $H$ must contain some member of $\F$ as a (non necessarily spanning) subgraph.
}
\RRmotcle{Hypergraphe, Probl\`eme Minimum $\F$-Overlay, \np-compl\'etude, Complexit\'e param\'etr\'e}

\RRkeyword{Hypergraph, Minimum $\F$-Overlay Problem, \np-completeness, Fixed-parameter tractability}

\RRprojet{ABS, COATI}  

\RCSophia 

\begin{document}

\makeRR   

\tableofcontents
\newpage

\section{Introduction}


\subsection{Notation}

Most notations of this paper are standard. We now recall some of them, and we refer the reader to~\cite{Die12} for any undefined terminology. For a graph $G$, we denote by $V(G)$ and $E(G)$ its respective sets of \emph{vertices} and \emph{edges}. The \emph{order} of a graph $G$ is $|V(G)|$, while its \emph{size} is $|E(G)|$. By extension, for a hypergraph $H$, we denote by $V(H)$ and $E(H)$ its respective sets of \emph{vertices} and \emph{hyperedges}. For $p \in \N$, a  {\it $p$-uniform} hypergraph $H$ is a hypergraph such that $|S|=p$ for every $S \in E(H)$. Given a graph $G$, we say that a graph $G'$ is a \emph{subgraph} of $G$ if $V(G') \subseteq V(G)$ and $E(G') \subseteq E(G)$. We say that $G'$ is a \emph{spanning subgraph} of $G$ if it is a subgraph of $G$ such that $V(G') = V(G)$. Given $S \subseteq V(G)$, we denote by $G[S]$ the graph with vertex set $S$ and edge set $\{uv \in E(G) \mid u, v \in S\}$. We say that a graph $G'$ is an \emph{induced subgraph} of $G$ if there exists $S \subseteq V(G)$ such that $G' = G[S]$. Given $S \subseteq V(G)$, we say that an edge $uv \in E(G)$ is \emph{covered} by $S$ if $u \in S$ or $v \in S$, and we say that $uv \in E(G)$ is \emph{induced} by $S$ if $\{u,v\}\subseteq S$. An \emph{isolated} vertex of a graph is a vertex of degree $0$. Finally, for a positive integer $p$, let $[p] = \{1, \dots, p\}$.

\subsection{Definition of the \textsc{Minimum $\F$-Overlay} problem}

We define the problem investigated in this paper: \textsc{Minimum $\F$-Overlay}.
Given a fixed family of graphs ${\cal F}$ and an input hypergraph $H$, we say that a graph $G$ {\it overlays} ${\cal F}$ on $H$ if $V(G) = V(H)$ and for every hyperedge $S \in E(H)$, the subgraph of $G$ induced by $S$, $G[ S]$, has a spanning subgraph in ${\cal F}$.

Observe that if a graph $G$ overlays ${\cal F}$ on $H$, then the graph $G$ with any additional edges overlays ${\cal F}$ on $H$.
Thus, there exists a graph $G$ overlaying ${\cal F}$ on $H$ if and only if the complete graph on $|V(H)|$ vertices overlays ${\cal F}$ on $H$.
Note that the complete graph on $|V(H)|$ vertices overlays ${\cal F}$ on $H$ if and only if for every hyperedge $S \in E(H)$, there exists a graph in ${\cal F}$ with exactly $|S|$ vertices.
It implies that deciding whether there exists a graph $G$ overlaying $\F$ on $H$ can be done in polynomial time. Hence, otherwise stated, we will always assume that there exists a graph overlaying $\F$ on our input hypergraph $H$.
We thus focus on minimizing the number of edges of a graph overlaying ${\cal F}$ on $H$.

The {\it ${\cal F}$-overlay number} of a hypergraph $H$, denoted $\ov_{\cal F}(H)$, is the smallest size (\ie number of edges) of a graph overlaying ${\cal F}$ on $H$.

\computationalproblem{Minimum $\cal F$-Overlay}
{A hypergraph $H$, and an integer $k$.}
{$\ov_{F}(H) \leq k$?}


We also investigate a relaxed version of the problem, called \textsc{Minimum $\F$-Encompass} where we ask for a graph $G$ such that for every hyperedge $S \in E(H)$, the graph $G[ S]$ contains a (non necessarily spanning) subgraph in ${\cal F}$.
In an analogous way, we define the {\it ${\cal F}$-encompass number}, denoted $\enc_{\cal F}(H)$, of a hypergraph $H$.

\computationalproblem{Minimum ${\cal F}$-Encompass}
{A hypergraph $H$, and an integer $k$.}
{$\enc_{F}(H) \leq k$?}

Observe that the {\sc Minimum ${\cal F}$-Encompass} problems are particular cases of {\sc Minimum ${\cal F}$-Overlay} problems.
Indeed, for a family ${\cal F}$ of graphs, let $\tilde{{\cal F}}$ be the family of graphs containing an element of ${\cal F}$ as a subgraph.
Then {\sc Minimum ${\cal F}$-Encompass} is exactly {\sc Minimum $\tilde{\cal F}$-Overlay}.


Throughout the paper, we will only consider graph families $\F$ whose \textsc{$\F$-Recognition} problem\footnote{The \textsc{$\F$-Recognition} problem asks, given a graph $F$, whether $F \in \F$.} is in $\np$. This assumption implies that \textsc{Minimum $\F$-Overlay} and \textsc{Minimum $\F$-Encompass} are in $\np$ as well (indeed, a certificate for both problems is simply a certificate of the recognition problem for every hyperedge). In particular, it is not necessary for the recognition problem to be in $\poly$ as it can be observed from the family $\F_{Ham}$ of Hamiltonian graphs: the \textsc{$\F$-Recognition} problem is $\np$-hard, but providing a spanning cycle for every hyperedge is a polynomial certificate and thus belongs to $\np$.

\subsection{Related work and applications}

\textsc{Minimum $\F$-Overlay} allows us to model lots of interesting combinatorial optimization problems of practical interest, as we proceed to discuss.

Common graph families $\F$ are the following:
connected graphs (and more generally, $\ell$-connected graphs), Hamiltonian graphs, graphs having a universal vertex (\ie having a vertex adjacent to every other vertex). When the family is the set of all connected graphs, then the problem is known as
{\sc Subset Interconnection Design}, {\sc Minimum Topic-Connected Overlay} or {\sc Interconnection Graph Problem}. It has been studied by several communities in the context of designing vacuum systems \cite{DK95,DM88}, scalable overlay networks \cite{CMT+07,HHI+12,OnRi11}, reconfigurable interconnection
networks~\cite{FHW08,FanWu08}, and, in variants, in the context of inferring a most
likely social network \cite{AAR10}, determining winners of combinatorial auctions \cite{CDS04}, as well as drawing hypergraphs \cite{BCP+11,KMN14,JoPo87,KoSt03}.

As an illustration, we explain in detail the importance of such inference problems for fundamental questions on structural biology~\cite{ACCC15}.
A major problem is the characterization of low resolution structures of macro-molecular assemblies.
To attack this very difficult question, one has to determine the
plausible contacts between the subunits of an assembly, given the lists of subunits
involved in all the complexes.
We assume that the composition, in terms of individual subunits, of
selected complexes is known.
Indeed, a given assembly can be chemically
split into complexes by manipulating chemical conditions.
This problem can be formulated as a \textsc{Minimum $\F$-Overlay} problem, where vertices represent the subunits and hyperedges are the complexes. In this setting, an edge between two vertices represents a contact between two subunits.

\vspace{0.06cm}
\noindent
\begin{minipage}[c]{.6\linewidth}
Hence, the considered family $\F$ is the family of all trees: we want the complexes to be connected. 
Note that the minimal connectivity assumption avoids speculating on the exact
(unknown) number of contacts.
Indeed, due to volume exclusion constraints, a given subunit cannot
contact many others.
The figure depicts a simple assembly
composed of four complexes (hyperedges) and an optimal solution.
We can also add some other constraints to the family such as `bounded maximum degree': a subunit (e.g. a protein) cannot be connected to many other subunits (vertices).
\end{minipage}
\hspace{0.185cm}
\begin{minipage}[c]{.36\linewidth}
\vspace{-0.13cm}
\includegraphics[width=1\textwidth]{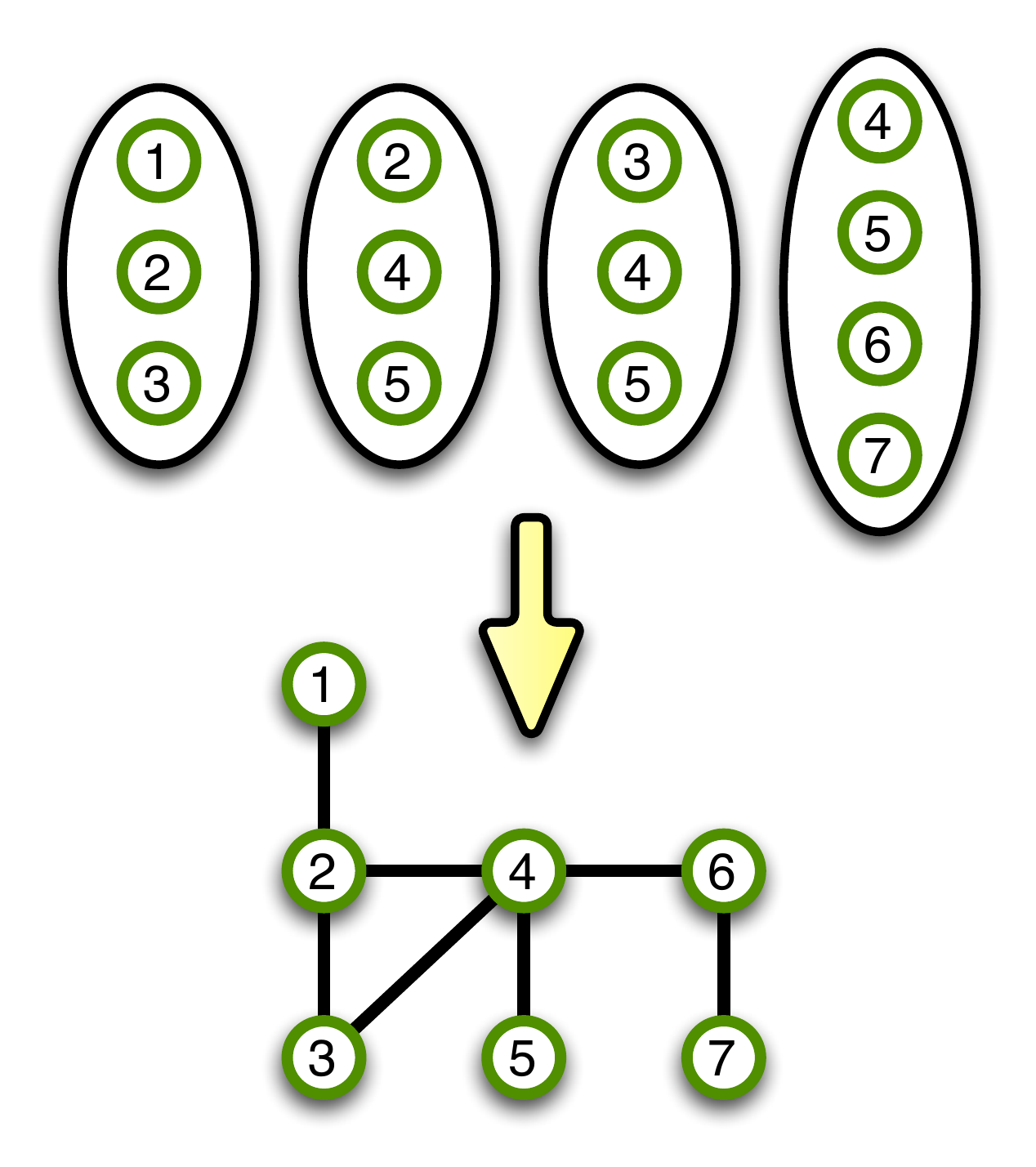}
\end{minipage}




\subsection{Our contributions}

In Section~\ref{sec:NP}, we prove a strong dichotomy result regarding the polynomial vs. \np-hard status with respect to the considered family $\F$. Roughly speaking, we show that the easy cases one can think of (\textit{e.g.} containing only edgeless and complete graphs) are the only families giving rise to a polynomial problem: all others are $\np$-complete. In particular, it implies that the \textsc{Minimum Connectivity Inference} problem is $\np$-hard in $p$-uniform hypergraphs, which generalizes previous results.
In Section~\ref{sec:FPT}, we then investigate the parameterized complexity of the problem and give similar sufficient conditions on $\F$ that gives rise to $\W[1]$-hard, $\W[2]$-hard or $\fpt$ problems.
 This yields an \fpt/$\W[1]$-hard dichotomy for \textsc{Minimum $\F$-Encompass}.


\section{Complexity dichotomy}
\label{sec:NP}

In this section, we prove a dichotomy between families of graphs ${\cal F}$ such that
 \textsc{Minimum $\F$-Overlay} is polynomial-time solvable, and families of graphs ${\cal F}$ such that
 \textsc{Minimum $\F$-Overlay} is \np-complete.

Given a family of graphs $\F$ and a positive integer $p$, let $\F_p = \{F \in \F: |V(F)|=p\}$.
We denote by $K_p$ the complete graph on $p$ vertices, and by $\overline{K_p}$ the edgeless graph on $p$ vertices.

\begin{theorem}\label{thm:main}
Let $\F$ be a family of graphs. If, for every $p > 0$, either $\F_p = \emptyset$ or $\F_p = \{K_p\}$ or $\overline{K_p} \in \F_p$, then \textsc{Minimum $\F$-Overlay} is polynomial-time solvable. Otherwise, it is $\np$-complete.
\end{theorem}


Let us first prove the first part of this theorem.

\begin{theorem}\label{thm:poly}
Let $\F$ be a set of graphs. If, for every $p > 0$, either $\F_p = \emptyset$ or $\F_p = \{K_p\}$ or $\overline{K_p} \in \F_p$, then \textsc{Minimum $\F$-Overlay} is polynomial-time solvable.
\end{theorem}
\begin{proof}
Let $I_0, I_1$, and $I_2$ be the sets of positive integers  $p$ such that, respectively, ${\cal F}_p = \emptyset$, $\K_p\in {\cal F}_p$, and
${\cal F}_p=\{K_p\}$.
The following trivial algorithm solves {\sc Minimum ${\cal F}$-Overlay} in polynomial time.
Let $H$ be a hypergraph.
If it contains a hyperedge whose size is in $I_0$, return `No'. If not, then for every hyperedge $S$ whose size is in $I_2$, add the ${|S| \choose 2}$ edges with endvertices in $S$.
If the number of edges of the resulting graph (which is a minimum solution) is at most $k$, return `Yes'. Otherwise return `No'.
\end{proof}


The \np-complete part requires more work. We need to prove that if there exists $p > 0$ such that $\F_p \neq \emptyset$, $\F_p \neq \{K_p\}$, and $\overline{K_p} \notin \F_p$, then \textsc{Minimum $\F$-Overlay} is $\np$-complete. Actually, it is sufficient to prove the following:
\begin{theorem}\label{thm:gen-1}
Let $p > 0$, and $\F_p$ be a non-empty set of graphs with $p$ vertices such that $\F_p \neq \{K_p\}$ and $\overline{K_p} \notin \F_p$. Then \textsc{Minimum $\F_p$-Overlay} is $\np$-complete (when restricted to $p$-uniform hypergraphs).
\end{theorem}

\subsection{Prescribing some edges}

A natural generalization of  {\sc Minimum $\cal F$-Overlay}  is to prescribe a set $E$ of edges to be in the
graph overlaying $\cal F$ on $H$.
We denote by $\ov_{\cal F}(H; E)$ the minimum number of edges of a graph $G$ overlaying $\cal F$ on $H$ with $E\subseteq E(G)$.

\smallskip

\computationalproblem{Prescribed Minimum $\cal F$-Overlay}
{A hypergraph $H$, an integer $k$, and a set $E\subseteq {V(H) \choose 2}$.}
{$\ov_{\cal F}(H;E) \leq k$?}

In fact, in terms of computational complexity, the two problems {\sc Minimum $\cal F$-Overlay} and {\sc Prescribed Minimum $\cal F$-Overlay} are equivalent.

\begin{theorem}\label{thm:equiv}
Let ${\cal F}$ be  a (possibly infinite) class of graphs.
Then  {\sc Minimum $\cal F$-Overlay} and {\sc Prescribed Minimum $\cal F$-Overlay} are polynomially equivalent.
\end{theorem}
\begin{proof}
An instance $(H, k)$ of {\sc Minimum $\cal F$-Overlay} is clearly equivalent to the instance $(H, k, \emptyset)$ of {\sc Prescribed Minimum $\cal F$-Overlay}.
This gives an easy polynomial reduction from {\sc Minimum $\cal F$-Overlay} to {\sc Prescribed Minimum $\cal F$-Overlay}.

\medskip

We now give a polynomial reduction from {\sc Prescribed Minimum $\cal F$-Overlay} to {\sc Minimum $\cal F$-Overlay}. Let us denote by ${\cal F}_p$ the set of graphs of $\cal F$ with order $p$.
Clearly, if ${\cal F}_p =\emptyset$ or $\K_p \in {\cal F}_p$ for every positive integer $p$, then both  {\sc Minimum $\cal F$-Overlay} and {\sc Prescribed Minimum $\cal F$-Overlay} are polynomial-time solvable.

We may assume henceforth  that there exists $p$ such that ${\cal F}_p \neq \emptyset$ and $\K_p \notin {\cal F}_p$.
Let $F$ be an element of ${\cal F}_p$ with the minimum number of edges. Observe that $|E(F)| \geq 1$.


Let $(H, k, E)$ be an instance of {\sc Prescribed Minimum $\cal F$-Overlay}. For every edge $e=u_ev_e\in E$, we add a set $X_e$ of $|V(F)| - 2$ new vertices and the hyperedge $S_e=X_e\cup \{u_e, v_e\}$.
Let $H'$ be the hypergraph defined by $V(H')=V(H)\cup \bigcup_{e\in E} X_e$ and $E(H')=E(H)\cup \{S_e \mid e\in E\}$.
We shall prove that $\ov_{\cal F}(H') = \ov_{\cal F}(H;E) + |E|(|F|-1)$.

\medskip
Suppose first that there is a graph $G$ overlaying $\cal F$ on $H$ with $E\subseteq E(G)$ and $|E(G)|\leq k$.
For any edge $\in E$, let $F_e$ be a copy of $F$ with vertex set $S_e$ such that $e\in E(F_e)$.
Such a $F_e$ exists because $F$ is non-empty.
Let $G'$ be the graph with vertex set $V(H')$ and edge set $E(G)\cup \bigcup_{e\in E} E(F_e)$.
Clearly, $G'$ is a graph overlaying $\cal F$ on $H'$ with $k+|E|(|F|-1)$ edges.

\medskip
Reciprocally, assume that  $\ov_{\cal F}(H') \leq k + |E|(|F|-1)$. Let $G'$ be a graph overlaying $\cal F$ on $H'$ of size at most $k+|E|(|F|-1)$ whose number of edges in $E$ is maximum.

We claim that $E\subseteq E(G')$.
Suppose not. Then there is an edge $e\in E\setminus E(G')$. Let $F_e$
be be a copy of $F$ with vertex set $S_e$ such that $e\in E(F_e)$.
Since the vertices of $X_e$ are only in the hyperedge $S_e$ of $H'$, replacing the edges of $G'[ S_e]$ by $E(F_e)$ in $G'$
results in a graph overlaying $\cal F$ on $H'$ of size $k+|E|(|F|-1)$ containing one more edge in $E$, a contradiction.
This proves the claim.

Let $G$ be the graph with vertex set $V(H)$ and edge set $E(H')\cap {V(H) \choose 2}$. Clearly, $G$ is a graph overlaying $\cal F$ on $H$, and by the above claim $E\subseteq E(G)$.
Now for every $e\in E$, $G'[ S_e]$ contains (at least) $|F|$ edges and only one of them is in $E(G)$.
Therefore, $|E(G)|\leq |E(G')|-|E|(|F|-1) \leq k$.
\end{proof}

\subsection{Hard sets}

A set ${\cal F}_p$ of graphs of order $p$ is {\it hard} if there is a graph $J$ of order $p$ and two distinct non-edges $e_1, e_2$ of $J$ such that
\begin{itemize}
\item[$\bullet$] no subgraph of $J$ is in ${\cal F}_p$ (including $J$ itself),
\item[$\bullet$] $J\cup e_1$ has a subgraph in  ${\cal F}_p$ and $J\cup e_2$ has a subgraph in  ${\cal F}_p$.
\end{itemize}
The graph $J$ is called the {\it hyperedge graph of ${\cal F}_p$} and $e_1$ and $e_2$ are its two {\it shifting non-edges}.

\begin{lemma}
\label{lem:hard}
Let $p \ge 3$ and ${\cal F}_p$ be a set of graphs of order $p$.
If ${\cal F}_p$ is hard, then
{\sc Prescribed Minimum ${\cal F}_p$-Overlay} is \np-complete.
\end{lemma}
\begin{proof}
We present a reduction from {\sc Vertex Cover}.
Let $J$ be the hyperedge graph of ${\cal F}_p$ and $e_1, e_2$ its shifting non-edges.
We distinguish two cases depending on whether $e_1$ and $e_2$ are disjoint or not.
The proofs of both cases are very similar.

\medskip

\noindent \underline{Case 1:} $e_1$ and $e_2$ intersect.
Let $G$ be a graph.
Let $H_G$ be the hypergraph constructed as follows.
\begin{itemize}
\item[$\bullet$] For every vertex $v \in V(G)$ add two vertices $x_v$, $y_v$.
\item[$\bullet$] For every edge $e=uv$, add a vertex $z_e$ and three disjoint sets $Z_e$, $Y^e_u$, and $Y^e_v$ of size $p-3$.
\item[$\bullet$] For every edge $e=uv$, create three hyperedges $Z_e\cup \{z_e, y_u, y_v\}$, $Y^e_u\cup \{x_u, y_u, z_e\}$, and
$Y^e_v\cup \{x_v, y_v, z_e\}$.
\end{itemize}


We select forced edges as follows: for every edge $e=uv\in E(G)$, we force the edges of a copy of $J$ on $Z_e\cup \{z_e, y_u, y_v\}$ with shifting non-edges $z_ey_u$ and $z_ey_v$, we force the edges of a copy of $J$ on $Y^e_u\cup \{z_e, y_u, x_u\}$ with shifting non-edges $y_uz_e$ and $y_ux_u$, and
we force the edges of a copy of $J$ on $Y^e_v\cup \{z_e, y_v, x_v\}$ with shifting non-edges $y_vz_e$ and $y_vx_v$.

We shall prove that $\ov_{{\cal F}_p}(H_G) = |E| +  \vc(G) + |E(G)|$, which yields the result. 
Here, $\vc(G)$ denotes the size of a minimum vertex cover of $G$.

\smallskip

Consider first a minimum vertex cover $C$ of $G$.
For every edge $e\in E(G)$, let $s_e$ be an endvertex of $e$ that is not in $C$ if such vertex exists, or any endvertex of $e$ otherwise.
Set $E_G=E\cup \{x_vy_v \mid v\in C\} \cup \{z_{e}y_{s_e} \mid e \in E(G)\}$.
One can easily check that $(V_G, E\cup E_G)$ overlays ${\cal F}_p$ on $H_G$. Indeed, for every hyperedge $S$ of $H_G$,  at least one of the shifting non-edges of its forced copy of $J$ is an edge of $E\cup E_G$.
Therefore $\ov_{{\cal F}_p}(H_G) \leq |E| + |E_G| = |E| + \vc(G) + |E(G)|$.

\medskip

Now, consider a minimum-size graph  $(V_G, E\cup E_G)$ overlaying ${\cal F}_p$ on $H_G$ and maximizing the edges of the form $x_uy_u$.
Let $e=uv\in E(G)$.
Observe that the edge $y_uy_{v}$ is contained in a unique hyperedge, namely  $Z_e\cup\{z_{e}, y_u, y_v\}$.
Therefore, free to replace it (if it is not in $E$) by $z_{e}y_v$, we may assume that $y_uy_v\notin E_G$.
Similarly, we may assume that the edges $x_uz_{e}$ and $x_vz_{e}$ are not in $E_G$, and that no edge with an endvertex in $Y^e_u\cup Y^e_v\cup Z_e$ is in $E_G$.
Furthermore, one of $x_uy_u$ and $x_vy_v$ is in $E_G$.
Indeed, if $\{x_uy_u, x_vy_v\} \cap E_G=\emptyset$, then $\{y_uz_{e}, y_vz_{e}\} \subseteq E_G$ because $E_G$ contains an edge included in every hyperedge. Thus replacing $y_uz_{e}$ by $x_uy_u$ results in another graph overlaying ${\cal F}_p$ on $H_G$ with one more edge of type $x_u y_u$ than the chosen one, a contradiction.

Let $C=\{ u \mid x_uy_u \in E_G\}$.
By the above property, $C$ is a vertex cover of $G$, so $|C|\geq \vc(G)$.
Moreover, $E_G$ contains an edge in every hyperedge $Z_e\cup \{z_{e}, y_u, y_v\}$, and those $|E(G)|$ edges are not in $\{x_uy_u \mid u\in V(G)\}$.
Therefore $|E_G|\geq |C| + |E(G)|\geq \vc(G) + |E(G)|$.

\smallskip

\noindent \underline{Case 2:} $e_1$ and $e_2$ are disjoint, say $e_1=x_1y_1$ and $e_2=x_2y_2$ (thus $p \ge 4$).
Let $G$ be a graph.
Let $H_G$ be the hypergraph constructed as follows.
\begin{itemize}
\item[$\bullet$] For every vertex $v \in V(G)$, add two vertices $x_v$, $y_v$.
\item[$\bullet$] For every edge $e=uv$, add four vertices $x_u^e, y_u^e, x_v^e, y_v^e$ and three disjoint sets $Z_e$, $Y^e_u$ and $Y^e_v$ of size $p-4$.
\item[$\bullet$] For every edge $e=uv$, create three hyperedges $Z_e\cup \{x_u^e, y_u^e, x_v^e, y_v^e\}$, $Y^e_u\cup \{x_u, y_u, x_u^e, y_u^e\}$, and
$Y^e_v\cup \{x_v, y_v,x_v^e, y_v^e\}$.
\end{itemize}

We select forced edges as follows: for every edge $e=uv\in E(G)$, we force the edges of a copy of $J$ on $Z_e\cup \{x_u^e, y_u^e, x_v^e, y_v^e\}$ with shifting non-edges $x_u^e, y_u^e$ and $x_v^e, y_v^e$, we force the edges of a copy of $J$ on $Y^e_u\cup \{x_u, y_u, x_u^e, y_u^e\}$  with shifting non-edges $x_uy_u$ and $x_u^e, y_u^e$, and
we force the edges of a copy of $J$ on $Y^e_v\cup \{x_v, y_v,x_v^e, y_v^e\}$  with shifting non-edges $x_vy_v$ and $x_v^e, y_v^e$.

We shall prove that $\ov_{{\cal F}_p}(H_G) = |E| +  \vc(G) + |E(G)|$, which yields the result.

\smallskip

Consider first a minimum vertex cover $C$ of $G$.
For every edge $e\in E(G)$, let $s_e$ be an endvertex of $e$ that is not in $C$ if one such vertex exists, or any endvertex of $e$ otherwise.
Set $E_G=E\cup \{x_vy_v \mid v\in C\} \cup \{x^e_{s_e}y^e_{s_e} \mid e \in E(G)\}$.
One can easily check that $(V_G, E\cup E_G)$ overlays ${\cal F}_p$ on $H_G$.
Indeed, for every hyperedge $S$ of $H_G$,  at least one of the shifting non-edges of its forced copy of $J$ is an edge of $E\cup E_G$.
Therefore $\ov_{{\cal F}_p}(H_G) \leq |E| + |E_G| = |E| + \vc(G) + |E(G)|$.

\smallskip

Now, consider a minimum-size graph  $(V_G, E\cup E_G)$ overlaying ${\cal F}_p$ on $H_G$ and maximizing the edges of the form $x_uy_u$.
Let $e=uv\in E(G)$.
Observe that the edge $x_ux_u^e$ is contained in a unique hyperedge, namely  $Y^e_u\cup \{x_u, y_u, x_u^e, y_u^e\}$.
Therefore, free to replace it (if it is not in $E$) by $x_uy_u$, we may assume that $x_ux_u^e\notin E_G$.
Similarly, we may assume that the edges $x_uy^e_u$, $y_ux^e_u$, $y_uy^e_u$, $x_vx_v^e$ $x_vy^e_v$, $y_vx^e_v$, $y_vy^e_v$,
 $x^e_ux_v^e$ $x^e_uy^e_v$, $y^e_ux^e_v$, and  $y^e_uy^e_v$
are not in $E_G$, and that no edge with an endvertex in $Y^e_u\cup Y^e_v\cup Z_e$ is in $E_G$.
Furthermore, one of $x_uy_u$ and $x_vy_v$ is in $E_G$.
Indeed, if $\{x_uy_u, x_vy_v\} \cap E_G=\emptyset$, then $\{x^e_uy^e_u, x^e_vy^e_v\} \subseteq E_G$ because $E_G$ contains an edge included in every hyperedge. Thus replacing $x^e_uy^e_u$ by $x_uy_u$ results in another graph overlaying ${\cal F}_p$ on $H_G$ with one more edge of type $x_u y_u$ than the chosen one, a contradiction.

Let $C=\{ u \mid x_uy_u \in E_G\}$.
By the above property, $C$ is a vertex cover of $G$, so $|C|\geq \vc(G)$.
Moreover, $E_G$ contains an edge in every hyperedge $Z_e\cup \{x_u^e, y_u^e, x_v^e, y_v^e\}$, and those $|E(G)|$ edges are not in $\{x_uy_u \mid u\in V(G)\}$.
Therefore $|E_G|\geq |C| + |E(G)|\geq \vc(G) + |E(G)|$.
\end{proof}

Let ${\cal F}_p$ be a set of graphs of order $p$.
It is {\it free}
if there are no two distinct elements of ${\cal F}_p$ such that one is a subgraph of the other.
The {\it core} of ${\cal F}_p$ is the free set of graphs $F$ having no proper subgraphs in ${\cal F}_p$.
It is easy to see that $\F_p$ is overlayed by a hypergraph if and only if its core does.
Henceforth, we may restrict our attention to free sets of graphs.

\begin{lemma}\label{lem:isol-hard}
Let ${\cal F}_p$ be a free set of graphs of order $p$.
If a graph $F$ in ${\cal F}_p$ has an isolated vertex and a vertex of degree $1$, then $\F_p$ is hard.
\end{lemma}
\begin{proof}
Let $z$ be an isolated vertex of $F$, $y$ a vertex of degree $1$, and $x$ the neighbor of $y$ in $F$.
The graph $J=F\setminus xy$ contains no element of ${\cal F}_p$ because ${\cal F}_p$ is free.
Moreover $J\cup xy$ and $J\cup yz$ are isomorphic to $F$.
Hence $J$ is a hyperedge graph of ${\cal F}_p$.
Thus, by Lemma~\ref{lem:hard}, {\sc Prescribed Minimum ${\cal F}_p$-Overlay} is \np-complete. 
\end{proof}

The {\it star of order $p$}, denoted by $S_p$, is the graph of order $p$ with $p-1$ edges incident to a same vertex.

\begin{lemma}\label{lem:star-hard}
Let $p\geq 3$ and
let ${\cal F}_p$ be a free set of graphs of order $p$ containing a subgraph of the star $S_p$ different from $\K_p$.
Then ${\cal F}_p$ is hard.
\end{lemma}
\vspace{-5mm}
\begin{center}
  \begin{tikzpicture}[xscale=.5]
    \draw node[fill,circle,scale=.5,label={0:$c_1$}] (c) at (0,0) {};
    \foreach \i in {2,...,7} {
      \draw node[fill,circle,scale=.5,label={90:$c_\i$}] (c\i) at (\i*.8-4.5*.8,1) {};
      \draw (c) -- (c\i);
    }
    \draw node at (0,-.6) {$S_8$};
  \end{tikzpicture}
  \hspace{1cm}
  \begin{tikzpicture}[xscale=.5]
    \draw node[fill,circle,scale=.5,label={0:$b$}] at (0,1.8) {};
    \draw node[fill,circle,scale=.5,label={180:$a_1$}] (a1) at (-1,0) {};
    \draw node[fill,circle,scale=.5,label={0:$a_2$}]   (a2) at ( 1,0) {};
    \foreach \i in {1,...,5} {
      \draw node[fill,circle,scale=.5,label={90:$c_\i$}] (c\i) at (\i*.8-3*.8,1) {};
      \draw (a1) -- (c\i) -- (a2);
    }
    \draw (-1,0) -- (1,0);
    \draw node at (0,-.6) {$Q_8$};
  \end{tikzpicture}
\end{center}
\begin{proof}
Let $S$ be the non-empty subgraph of $S_p$ in ${\cal F}_p$.
If $S\neq S_p$, then $S$ has an isolated vertex and a vertex of degree $1$, and so ${\cal F}_p$ is hard by Lemma~\ref{lem:isol-hard}. We may assume henceforth that $S_p\in {\cal F}_p$.

Let $Q_p$ be the graph with $p$ vertices $\{a_1,a_2, b, c_1, \dots , c_{p-3}\}$ and edge set
$\{a_1a_2\} \cup \{a_ic_j \mid 1\leq  i \leq 2, 1\leq j \leq p-3\}$.
Observe that $Q_p$ does not contain $S_p$ but $Q_p\cup a_1b$ and $Q_p\cup a_2b$ do.
If ${\cal F}_p$ contains no subgraph of $Q_p$, then ${\cal F}_p$ is hard.
So we may assume that ${\cal F}_p$ contains a subgraph of $Q_p$.

Let $Q$ be the subgraph of $Q_p$ in ${\cal F}_p$ that has the minimum number of triangles.
If $Q$ has a degree $1$ vertex, then ${\cal F}_p$ is hard by Lemma~\ref{lem:isol-hard}.
Henceforth we may assume that $Q$ has no vertex of degree $1$. So, without loss of generality,
there exists $q$ such that $E(Q)=\{a_1a_2\} \cup \{a_ic_j \mid 1\leq  i \leq 2, 1\leq j \leq q\}$.

Let $R=(Q\setminus a_1c_1)\cup a_2b$.
Observe that $R\cup a_1c_1$ and $R\cup a_1b$ contain $Q$.
If ${\cal F}_p$ contains no subgraph of $R$, then ${\cal F}_p$ is hard.
So we may assume that ${\cal F}_p$ contains a subgraph $R'$ of $R$.
But ${\cal F}_p$ contains no subgraph of $Q$ because it is free, so both $a_2c_1$ and $a_2b$ are in $R'$.
In particular, $c_1$ and $b$ have degree $1$ in $R'$.

 Let $T=(Q\setminus a_1c_1)$. It is a proper subgraph of $Q$, so ${\cal F}_p$ contains no subgraph of $T$, because ${\cal F}_p$ is free.
 Moreover $T\cup a_1c_1=Q$ is in ${\cal F}_p$ and $T\cup a_2b=R$ contains $R'\in {\cal F}_p$.
 Hence ${\cal F}_p$ is hard.
 \end{proof}

\subsection{Proof of Theorem~\ref{thm:gen-1}}


For convenience, instead of proving Theorem~\ref{thm:gen-1}, we prove the following statement, which is equivalent by Theorem~\ref{thm:equiv}.

\begin{theorem}\label{thm:gen}
Let ${\cal F}_p$ be a non-empty set of graphs of order $p$.
{\sc Prescribed Minimum ${\cal F}_p$-Overlay} is \np-complete if $\K_p\in {\cal F}_p$ or ${\cal F}_p=\{K_p\}$.
\end{theorem}

\begin{proof}
We proceed by induction on $p$, the result holding trivially when $p=1$ and $p=2$. Assume now that $p\geq 3$. Without loss of generality, we may assume that ${\cal F}_p$ is a free set of graphs.

A {\it hypograph} of a graph $G$ is an induced subgraph of $G$ of order $|G|-1$. In other words, it is a subgraph obtained by removing a vertex from $G$.
Let ${\cal F}^-$ be the set of hypographs  of elements of ${\cal F}_p$.

\medskip

If ${\cal F}^-=\{K_{p-1}\}$, then necessarily ${\cal F}_p=\{K_p\}$, and {\sc Prescribed Minimum ${\cal F}_p$-Overlay} is trivially polynomial-time solvable.

\medskip

If ${\cal F}^-\neq \{K_{p-1}\}$ and $\K_{p-1}\notin {\cal F}^-$, then {\sc Prescribed Minimum ${\cal F}^-$-Overlay} is \np-complete by the induction hypothesis.
We shall now reduce this problem to {\sc Prescribed Minimum ${\cal F}_p$-Overlay}.
Let $(H^-,k^-, E^-)$ be an instance of {\sc Prescribed Minimum ${\cal F}^-$-Overlay}.
For every hyperedge $S$ of $H^-$, we create a new vertex $x_S$ and the hyperedge $X_S=S\cup \{x_S\}$.
Let $H$ be the hypergraph defined by $V(H)=V(H^-)\cup \bigcup_{S\in E(H^-)} x_S$ and $E(H)= \{X_S \mid S\in E(H^-)\}$.
We set $E=E^-\cup \bigcup_{S\in E(H^-)} \{x_Sv \mid v\in S\}$. 

Let us prove that $\ov_{{\cal F}_p}(H;E) = \ov_{{\cal F}^-}(H^-;E^-) + (p-1)\cdot |S|$.
Clearly, if $G^-=(V(H^-), F^-)$ overlays ${\cal F}^-$, then $G=(V(H), F^-\cup \bigcup_{S\in E(H^-)} \{x_Sv \mid v\in S\} )$ overlays ${\cal F}_p$.
Hence $\ov_{{\cal F}_p}(H;E) \leq \ov_{{\cal F}^-}(H^-;E^-) + (p-1)\cdot |S|$.
Reciprocally, assume that $G$ overlays ${\cal F}_p$. Then for each hyperedge $S$ of $H^-$, the graph $G[ X_S] \in {\cal F}_p$, and so
 $G[ S ] \in {\cal F}^-$. Therefore, setting the graph  $G^-=G[ V(H^-)]$ overlays ${\cal F}^-$. Moreover $E(G)\setminus E(G^-)= \bigcup_{S\in E(H^-)} \{x_Sv \mid v\in S\}$. Hence $\ov_{{\cal F}_p}(H;E) \geq \ov_{{\cal F}^-}(H^-;E^-) + (p-1)\cdot |S|$.

\medskip

Assume now that $\K_{p-1}\in {\cal F}^-$. Then ${\cal F}_p$ contains a subgraph of the star $S_p$.
If ${\cal F}_p$ contains $\K_{p}$, then {\sc Prescribed Minimum ${\cal F}_p$-Overlay} is trivially polynomial-time solvable.
Henceforth, we may assume that ${\cal F}_p$ contains a non-empty subgraph of $S_p$.
Thus, by Lemma~\ref{lem:star-hard}, ${\cal F}_p$ is hard, and so by Lemma~\ref{lem:hard},   {\sc Prescribed Minimum ${\cal F}_p$-Overlay} is \np-complete.
\end{proof}

\section{Parameterized analysis}
\label{sec:FPT}

We now focus on the parameterized complexity of our problems. A \emph{parameterization} of a decision problem $Q$ is a computable function $\kappa$ that assigns an integer $\kappa(I)$ to every instance $I$ of the problem. We say that $(Q, \kappa)$ is \emph{fixed-parameter tractable} ($\fpt$) if every instance $I$ can be solved in time $O(f(\kappa(I)) |I|^c)$, where $f$ is some computable function, $|I|$ is the encoding size of $I$, and $c$ is some constant independent of $I$ (we will sometimes use the $O^*(\cdot)$ notation that removes polynomial factors and additive terms).
Finally, the $\W[i]$-hierarchy of parameterized problems is typically used to rule out the existence of $\fpt$ algorithms, under the widely believed assumption that $\fpt \neq \W[1]$. For more details about fixed-parameter tractability, we refer the reader to the monograph of Downey and Fellows~\cite{DoFe13}.\\

Since \textsc{Minimum $\F$-Overlay} is $\np$-hard for most non-trivial cases, it is natural to ask for the existence of $\fpt$ algorithms. In this paper, we consider the so-called \emph{standard parameterization} of an optimization problem: the size of a solution. In the setting of our problems, this parameter corresponds to the number $k$ of edges in a solution. Hence, 
the considered parameter will always be $k$ in the remainder of this section.

Similarly to our dichotomy result stated in Theorem~\ref{thm:main}, we would like to obtain necessary and sufficient conditions on the family $\F$ giving rise to either an $\fpt$ or a $\W[1]$-hard problem. One step towards such a result is the following $\fpt$-analogue of Theorem~\ref{thm:poly}.

\begin{theorem}\label{thm:fpt}
Let $\F$ be a family of graphs. If there is a non-decreasing function $f : \N \rightarrow \N$ such that $\lim_{n \rightarrow + \infty} f(n) = + \infty$ and $|E(F)| \ge f(|V(F)|)$ for all $F \in \F$, then \textsc{Minimum $\F$-Overlay} is $\fpt$.
\end{theorem}
\begin{proof}
Let $g : \N \rightarrow \N$ be the function that maps every $k \in \N$ to the smallest integer $\ell$ such that $f(\ell) \ge k$. Since  $\lim_{n \rightarrow +\infty} f(n) = +\infty$, $g$ is well-defined.
If a hyperedge $S$ of a hypergraph $H$ is of size at least $g(k+1)$, then since $f$ is non-decreasing, $\ov_{\cal F}(H) > k$ and so the instance is negative. Therefore, we may assume that every hyperedge of $H$ has size at most $g(k)$. Applying a simple branching algorithm (see~\cite{DoFe13}) allows us to solve the problem in time $O^*(g(k)^{O(k)})$.
\end{proof}

Observe that if ${\cal F}$ is finite, setting $N=\max \{|E(F)| \mid F\in {\cal F}\}$, the function $f$ defined by $f(n) = 0$ for $n\leq N$ and $f(n) =n$ otherwise satisfies the condition of Theorem~\ref{thm:fpt}, and so  \textsc{Minimum $\F$-Overlay} is $\fpt$. Moreover, Theorem~\ref{thm:fpt} encompasses some interesting graph families. Indeed, if $\F$ is the family of connected graphs (resp. Hamiltonian graphs), then $f(n) = n-1$ (resp. $f(n) = n$) satisfies the required property. Other graph families include $c$-vertex-connected graphs or $c$-edge-connected graphs for any fixed $c \ge 1$, graphs of minimum degree at least $d$ for any fixed $d \ge 1$.
In sharp contrast, we shall see in the next subsection (Theorem~\ref{thm:Wone})
that if, for instance, $\F$ is the family of graphs containing a matching of size at least $c$, for any fixed $c \ge 1$, then the problem becomes $\W[1]$-hard 
(note that such a graph might have an arbitrary number of isolated vertices).

\subsection{Negative result}

In view of Theorem~\ref{thm:fpt}, a natural question is to know what happens for graph families not satisfying the conditions of the theorem. Although we were not able to obtain an exact dichotomy as in the previous section, we give sufficient conditions on $\F$ giving rise to problems that are unlikely to be $\fpt$ (by proving $\W[1]$-hardness or $\W[2]$-hardness).

An interesting situation is when $\F$ is \emph{closed by addition of isolated vertices}, \ie, for every $F \in \F$, the graph obtained from $F$ by adding an isolated vertex is also in $\F$. Observe that for such a family, \textsc{Minimum $\F$-Overlay} and \textsc{Minimum $\F$-Encompass} are equivalent, which is the reason that motivated us defining this relaxed version. We have the following result, which implies an $\fpt/\W[1]$-hard dichotomy for \textsc{Minimum $\F$-Encompass}.

\begin{theorem}\label{thm:Wone}
Let $\F$ be a fixed family of graphs closed by addition of isolated vertices. If $\K_p \in \F$ for some $p \in \N$, then \textsc{Minimum $\F$-Overlay} is $\fpt$. Otherwise, it is $\W[1]$-hard parameterized by $k$.
\end{theorem}
\begin{proof}
To prove the positive result, let $p$ be the minimum integer such that $\K_p \in \F$. Observe that no matter the graph $G$, for every hyperedge $S \in E(H)$, $G[S]$ will contain $\K_{|S|}$ as a spanning subgraph, which is in $\F$ whenever $|S| \ge p$ (recall that $\F$ is closed by addition of isolated vertices). Then, a simple branching algorithm allows us to enumerate all graphs (with at least one edge) induced by hyperedges of size at most $p-1$ in $O^*(p^{O(k)})$ time.

To prove the negative result, we use a recent result of Chen and Lin~\cite{ChLi16} stating that any constant-approximation of the parameterized \textsc{Dominating Set} is $\W[1]$-hard, which directly transfers to \textsc{Hitting Set}\footnote{Roughly speaking, each element of the universe represents a vertex of the graph, and for each vertex, create a set with the elements corresponding to its closed neighborhood.}.
For an input of \textsc{Hitting Set}, namely a finite set $U$ (called the \emph{universe}), and a family $\S$ of subsets of $U$, let $\tau(U, \S)$ be the minimum size of a set $K \subseteq U$ such that $K \cap S \neq \emptyset$ for all $S \in \S$ (such a set is called a \emph{hitting set}).
The result of Chen and Lin implies that the following problem is $\W[1]$-hard parameterized by $k$.\\

\computationalproblem{Gap$_{\rho}$ Hitting Set}
{A finite set $U$, a family $\S$ of subsets of $U$, and a positive integer $k$.}
{Decide whether $\tau(U, \S) \le k$ or $\tau(U, \S) > \rho k$.}

Let $F_{is}$ be a graph from $\F$ minimizing the two following criteria (in this order): number of non-isolated vertices, and minimum degree of non-isolated vertices. Let $r_{is}$ and $\delta_{is}$ be the respective values of these criteria, $n_{is} = |V(F_{is})|$, and $m_{is} = |E(F_{is})|$. We thus have $\delta_{is} \le r_{is}$.
Let $F_{e}$ be a graph in $\F$ with the minimum number of edges, and $n_e = |V(F_e)|$, $m_e=|E(F_e)|$.

Let $U, \S, k$ be an instance of \textsc{Gap$_{2r_{is}}$ Hitting Set}, with $U = \{u_1, \dots, u_n\}$. We denote by $H$ the hypergraph constructed as follows. Its vertex set is the union of:
\begin{itemize}
	\item[$\bullet$] a set $V_{is}$ of $r_{is}-1$ vertices;
	\item[$\bullet$] a set $V_{U} = \bigcup_{i=1}^n V^i$, where $V^i = \{v^i_1, \dots, v^i_{n_{is}-r_{is}+1}\}$; and
	\item[$\bullet$] for every $u, v \in V_{is}$, $u \neq v$, a set $V_{u, v}$ of $n_{e}-2$ vertices.
\end{itemize}
Then, for every $u, v \in V_{is}$, $u \neq v$, create a hyperedge $h_{u, v} = \{u, v\} \cup V_{u, v}$ and, for every set $S \in \S$, create the hyperedge $h_S =V_{is} \cup \bigcup_{i: u_i \in S} V^i$.
Finally, let $k' = {n_{is}-1 \choose 2} m_e + k \delta_{is}$. Since $\F$ is fixed, $k'$ is a function of $k$ only.

We shall prove that if $\tau(U, \S) \le k$, then $\ov_{\F}(H) \le k'$ and, conversely, if $\ov_{\F}(H) \le k'$, then $\tau(U, \S) \le 2r_{is} k$.

\medskip

Assume first that $U$ has a hitting set $K$ of size at most $k$. For every $u, v \in V_{is}$, $u \neq v$, add to $G$ the edges of a copy of $F_e$ on $h_{u, v}$ with $uv \in E(G)$. This already adds ${n_{is}-1 \choose 2} m_e $ edges to $G$ and, obviously, $G[h_{u, v}]$ contains $F_e$ as a subgraph. Now, for every $u_i \in K$, add all edges between $v^i_1$ and $\delta_{is}$ arbitrarily chosen vertices in $V_{is}$. Observe that for every $S \in \S$, $G[h_S]$ contains $F_{is}$ as a subgraph, and also $|E(G)| \le k'$.

\medskip
Conversely, let $G$ be a solution for \textsc{Minimum $\F$-Overlay} with at most $k'$ edges. Clearly, for all $u, v \in V_{is}$, $u \neq v$, $G[V_{u, v}]$ has at least $m_e$ edges, hence the subgraph of $G$ induced by $V(H) \setminus V_U$ has at least ${n_{is}-1 \choose 2} m_e $ edges, and thus the number of edges of $G$ covered by $V_u$ is at most $k \delta_{is}$. Let $K$ be the set of non-isolated vertices of $V_U$ in $G$, and $K' = \{u_i \mid v^i_j \in K$ for some $j \in \{1, \dots, n_{is}-r_{is}+1\}\}$. We claim that $K'$ is a hitting set of $(U, \S)$: indeed, for every $S \in \S$, $G[h_S]$ must contain some $F \in \F$ as a subgraph, but since $V_{is}$ is composed of $r_{is}-1$ vertices, and since $F_{is}$ is a graph from $\F$ with the minimum number $r_{is}$ of non-isolated vertices, there must exist $i \in \{1, \dots, n\}$ such that $u_i \in S$, and $j \in \{1, \dots, n_{is}-r_{is}+1\}$ such that $v^i_j \in h_S \cap K$, and thus $S \cap K' \neq \emptyset$. Finally, observe that $K$ is a set of non-isolated vertices covering $k \delta_{is}$ edges, and thus $|K| \le 2 k \delta_{is}$ (in the worst case, $K$ induces a matching), hence we have $|K'| \le |K| \le 2 k \delta_{is} \le 2 r_{is} k$, \ie $\tau(U, \S) \le 2 r_{is} k$, concluding the proof.
\end{proof}

It is worth pointing out that the idea of the proof of Theorem~\ref{thm:Wone} applies to broader families of graphs. Indeed, the required property `closed by addition of isolated vertices' forces $\F$ to contain all graphs $F_{is}+\K_i$ (where $+$ denotes the disjoint union of two graphs) for every $i \in \mathbb{N}$.
Actually, it would be sufficient to require the existence of a polynomial $p : \N \rightarrow \N$ such that for any $i \in \N$, we have $F_{is} + \K_{p(i)} \in \F$ (roughly speaking, for a set $S$ of the \textsc{Hitting Set} instance, we would construct a hyperedge with $F_{is}+\K_{p(|S|)}$ vertices).
Intuitively, most families of practical interest not satisfying such a constraint will fall into the scope of Theorem~\ref{thm:fpt}. Unfortunately, we were not able to obtain the dichotomy in a formal way.

Nevertheless, as explained before, this still yields an $\fpt/\W[1]$-hardness dichotomy for the \textsc{Minimum $\F$-Encompass} problem.

\begin{corollary}
Let $\F$ be a fixed family of graphs. If $\K_p \in \F$ for some $p \in \N$, then \textsc{Minimum $\F$-Encompass} is $\fpt$. Otherwise, it is $\W[1]$-hard parameterized by $k$.
\end{corollary}

We conclude this section with a stronger negative result than Theorem~\ref{thm:Wone}, but concerning a restricted graph family (hence both results are incomparable).

\begin{theorem}\label{thm:Wtwo}
Let $\F$ be a fixed graph family such that:
\begin{itemize}
	\item[$\bullet$]$\F$ is closed by addition of isolated vertices;
	\item[$\bullet$]$\K_p \notin \F$ for every $p \ge 0$; and
	\item[$\bullet$]all graphs in $\F$ have the same number of non-isolated vertices.
\end{itemize}
Then \textsc{Minimum $\F$-Overlay} is $\W[2]$-hard parameterized by $k$.
\end{theorem}
\begin{proof}
Let $F_{\delta}$ be a graph from $\F$ minimizing the minimum degree of non-isolated vertices.
Let $\delta$ be such a minimum degree and let $r$ be the number of non-isolated vertices of any graph $F$ of $\F$.
Let $n_{\delta} = |V(F_{\delta})|$ and $m_{\delta} = |E(F_{\delta})|$.
Let $F_{e}$ be a graph from $\F$ with the minimum number of edges, and $n_{e} = |V(F_e)|$, $m_{e} = |E(F_e)|$.

Let $U, \S, k$ be an instance of \textsc{Hitting Set}, with $U = \{u_1, \dots, u_n\}$. We denote by $H$ the hypergraph constructed as follows. Its vertex set is the union of:
\begin{itemize}
	\item[$\bullet$]a set $V_{\delta}$ of $r - 1$ vertices;
	\item[$\bullet$]a set $V_{U} = \bigcup_{i=1}^n V^i$, where $V^i = \{v^i_1, \dots, v^i_{n_{\delta} - r + 1}\}$;
	\item[$\bullet$]for every $u, v \in V_{\delta}$, $u \neq v$, a set $V_{u, v}$ of $n_{e}-2$ vertices.
\end{itemize}
Then, for every $u, v \in V_{\delta}$, $u \neq v$, create the hyperedge $h_{u, v} = \{u, v\} \cup V_{u, v}$, and, for every set $S \in \S$, create a hyperedge $h_S$ composed of $V_{\delta} \cup \bigcup_{i: u_i \in S} V^i$.
Finally, let $k' = {r - 1 \choose 2} m_{e} + k \delta$. Since $\F$ is fixed, $k'$ is a function of $k$ only.

We shall prove that $\tau(U, \S) \leq k$ if and only if $ov_{\F}(H) \le k'$.

\medskip

Assume first that $U$ has a hitting set $K$ of size at most $k$. For every $u, v \in V_{\delta}$, $u \neq v$, add to $G$ the edges of a copy of $F_e$ on $h_{u, v}$ with $uv \in E(G)$. This already adds ${n_{\delta}-1 \choose 2} m_e $ edges to $G$, and, obviously, $G[h_{u, v}]$ contains $F_e$ as a subgraph. Now, for every $u_i \in K$, add all edges between $v^{i}_{1}$ and $\delta$ vertices in $V_{\delta}$ (arbitrarily chosen). Observe that for every $S \in \S$, $G[h_S]$ contains $F_{\delta}$ as a subgraph, and also $|E(G)| \le k'$.

\medskip

Conversely, let $G = (V, E)$ be a solution for \textsc{Minimum $\F$-Overlay} with at most $k'$ edges maximizing $|E(G[V_{\delta}])|$. We claim that $G[V_{\delta}]$ is a clique. If not, let $u, v \in V_{\delta}$, $u \neq v$ such that $uv \notin E(G)$. Since $F_e$ is a graph from $\F$ inducing the minimum number of edges, and since all vertices of $V_{u, v}$ apart from $u$ and $v$ only belong to the hyperedge $h_{u, v}$, removing all edges from $G[V_{\delta}]$ to form a graph isomorphic to $F_e$ with $uv$ being an edge leads to a graph $G'$ with at most $k'$ edges and one more edge induced by $V_{\delta}$, a contradiction.
Then, observe that for every hyperedge $h_S$, there exists $v \in h_S \cap V_U$ such that $|N(v) \cap h_S| \ge \delta$ (recall that $|V_{\delta}| = r-1$). If $N(v) \cap V_U \cap h_S \neq \emptyset$, then remove from $G$ all edges between $v$ and any vertex of $h_S$, and add edges between $v$ and $\delta$ different arbitrarily chosen vertices form $V_{\delta}$. Since $G[V_{\delta}]$ is a clique, all hyperedges $h_{S'}$ containing the removed edges necessarily contain $v$ and thus contain $F_{\delta}$ as a subgraph. Hence this modification leads to a graph $G'$ inducing at most $k'$ edges which overlays $\F$ on $H$ and such that $N(v) \cap V_u \cap h_S = \emptyset$. We apply this rule whenever there exists $v \in h_S \cap V_U$ such that $N(v) \cap V_U \cap h_S \neq \emptyset$ and obtain a solution $G'$ with at most $k'$ edges such that for every hyperedge $h_S$, there exists $v_{j_S}^{i_S} \in h_S \cap V_U$ such that $|N(v^{i_S}_{j_S}) \cap V_{\delta}| = \delta$.
Let $X = \{v_{j_S}^{i_S} \mid S \in \S\}$. We have the following:
\begin{itemize}
	\item[$\bullet$]$X$ is a hitting set of hyperedges $\{h_S \mid S \in \S\}$ and, by construction, the set $X' = \{u_{i_S} \mid S \in \S\}$ is a hitting set of $(U, \S, k)$;
	\item[$\bullet$]since $G'$ has at most $k'$ edges, and $G'[V \setminus V_{U}]$ has ${r - 1 \choose 2} m_{e}$ edges, the number of edges covered by $X$ is at most $k \delta$; and
	\item[$\bullet$]for every $v \in X$, $|N_{G'}(v) \cap V_{\delta}| \geq \delta$.
\end{itemize}
Therefore, $X'$ is a hitting set of $(U, \S)$ of size at most $k$, which concludes the proof.
\end{proof}

Observe that the proof above is very similar to the one of Theorem~\ref{thm:Wone}. However, we could not reduce from the (non-approximated version of) \textsc{Hitting Set} for families $\F$ having different number of non-isolated vertices, for the following informal reasons:
\begin{itemize}
	\item[$\bullet$]the set $V_{\delta}$ must contain no more than $r-1$ vertices, where $r$ is the minimum number of non-isolated vertices of any graph from $\F$ (otherwise, since $V_{\delta}$ is forced to be a clique in any solution, any hyperedge $h_S$ would already contain some graph from $\F$).
	\item[$\bullet$]the graph $F^*$ chosen to be induced by hyperedges $h_S$ must be a graph with $r$ non-isolated vertices with a minimum degree.	
	\item[$\bullet$]it might be the case that $\F$ contains a graph $F'$ with more than $r$ non-isolated vertices but with a minimum degree smaller than the one of $F^*$. Thus, it would be possible to ``cheat'' and  put $F'$ in every hyperedge $h_S$: we would have more than one vertex of this graph in $V_U$ for each hyperedge, but they would cover in total less than $k \delta$ edges (hence we would be able to have a hitting set larger than $k$). However, the number of additional vertices we may win in the hitting set would only be of a linear factor of $k$. This is the reason why the reduction in the proof of Theorem~\ref{thm:Wone} is from the constant approximated version of \textsc{Hitting Set}.
\end{itemize}

\section{Conclusion and future work}


Naturally, the first open question is to close the gap between Theorems~\ref{thm:fpt} and \ref{thm:Wone} in order to obtain a complete $\fpt/\W[1]$-hard dichotomy for any family $\F$.

As further work, we are also interested in a more constrained version of the problem, in the sense that we may ask for a graph $G$ such that for every hyperedge $S \in E(H)$, the graph $G[S]$ belongs to $\F$ (hence, we forbid additional edges). The main difference between \textsc{Minimum $\F$-Overlay} and this problem, called \textsc{Minimum $\F$-Enforcement}, is that it is no longer trivial to test for the existence of a feasible solution (actually, it is possible to prove the $\np$-hardness of this existence test for very simple families, \eg when $\F$ only contains $P_3$, the path on three vertices).
We believe that a dichotomy result similar to Theorem~\ref{thm:main} for \textsc{Minimum $\F$-Enforcement} is an interesting challenging question, and will need a different approach than the one used in the proof of Theorem~\ref{thm:gen}.




\bibliographystyle{plain}

\bibliography{refs}

\end{document}